\RequirePackage{amsmath}
\documentclass[runningheads,envcountsame,a4paper]{llncs}
\usepackage[english]{babel}
\usepackage[utf8]{inputenc}
\usepackage{amsmath,amssymb}
\usepackage{multirow,epic,gastex,graphicx}
\usepackage[colorinlistoftodos]{todonotes}
\usepackage{verbatim}
\usepackage{rotating}
\DeclareSymbolFont{rsfscript}{OMS}{rsfs}{m}{n}
\DeclareSymbolFontAlphabet{\mathrsfs}{rsfscript}
\usepackage{complexity}

\DeclareMathOperator{\rt}{rt}

\title{An Extremal Series of\\Eulerian Synchronizing Automata}
\author{Marek Szyku{\l}a\inst{1}\thanks{Supported in part by the National Science Centre, Poland under project number 2015/17/B/ST6/01893.}
\and Vojt\v{e}ch Vorel\inst{2}\thanks{Research supported by the Czech Science Foundation grant GA14-10799S and the GAUK grant No. 52215.}
} 
\authorrunning{M. Szyku{\l}a, V. Vorel} 

\institute{Institute of Computer Science, University of Wroc\l aw, Joliot-Curie 15, \\Wroc\l aw, Poland, \\ \email{msz@cs.uni.wroc.pl}, 
\and Faculty of Mathematics and Physics, Charles University, Malostransk{\'e} n{\'a}m. 25, \\Prague, Czech Republic,\\ \email{vorel@ktiml.mff.cuni.cz}}

\begin{document}
\maketitle

\begin{abstract}
We present an infinite series of $n$-state Eulerian automata whose reset words have length at least $(n^2-3)/2$.
This improves the current lower bound on the length of shortest reset words in Eulerian automata.
We conjecture that $(n^2-3)/2$ also forms an upper bound for this class
and we experimentally verify it for small automata by an exhaustive computation.

\medskip
\textbf{Keywords:} Eulerian automaton, reset threshold, reset word, synchronizing automaton
\end{abstract}

\section{Introduction}

A complete deterministic finite automaton is \emph{synchronizing} if there exists a word whose action maps all states to a single one. Such words are called \emph{reset words}.
Synchronizing automata find applications in various fields such as robotics, coding theory, bioinformatics, and model-based testing.
Besides of these, synchronizing automata are of great theoretical interest, mainly because of the famous \v{C}ern\'{y} conjecture \cite{Cerny1964}, which is one of the most long-standing open problems in automata theory.
The conjecture states that each synchronizing $n$-state automaton has a reset word of length at most $(n-1)^2$.
The best known general upper bound on this length is $\frac{1}{6}n^3-\frac{1}{6}n-1$ for each $n\ge 4$. \cite{Pin1983OnTwoCombinatorialProblems}.
Surveys on the field can be found in~\cite{KariVolkov2013Handbook,Volkov2008Survey}.

Major research directions in this field
include proving the \v{C}ern\'{y} conjecture for special classes of automata or showing specific upper bounds for them.
For example, the \v{C}ern\'{y} conjecture has been positively solved for the classes of monotonic automata \cite{Ep1990}, circular automata \cite{Dubuc1998}, Eulerian automata \cite{Kari2003Eulerian}, aperiodic automata \cite{Tr2007Aperiodic}, one-cluster automata with a prime-length cycle \cite{Steinberg2011OneClusterPrime}, automata respecting intervals of a directed graph \cite{GK2013AutomataRespectingIntervals} (under an inductive assumption), and automata with a letter of rank at most $\sqrt[3]{6n-6}$ \cite{BerlinkovSzykula2015AlgebraicSynchronizationCriterion}.
Moreover, there are many improvements of upper bounds for important special classes, for example, generalized and weakly monotonic automata \cite{AV2005SynchronizingGeneralizedMonotonicAutomata,Volkov2009ChainOfPartialOrders}, one-cluster automata \cite{BBP2011QuadraticUpperBoundInOneCluster}, quasi-Eulerian and quasi-one-cluster automata \cite{Berlinkov2013QuasiEulerianOneCluster}, and decoders of finite prefix codes \cite{AV2005SynchronizingGeneralizedMonotonicAutomata,BerlinkovSzykula2015AlgebraicSynchronizationCriterion,BiskupPlandowski2009HuffmanCodes}.
On the other hand, several lower bounds have been established by showing extremal series of automata for particular classes 
\cite{AV2005SynchronizingGeneralizedMonotonicAutomata,BiskupPlandowski2009HuffmanCodes,Cerny1964,Gusev2013LowerBoundsForTheLengthOfResetWordsInEulerianAutomata}.
Still, for many classes the best known upper bound does not match the lower bound.

In this paper we deal with the class of Eulerian automata,
which is one of the most remarkable classes due to its properties with regard to synchronization. In particular, the lengths of shortest words extending subsets are at most $n-1$ for each $n$-state Eulerian automaton~\cite{Kari2003Eulerian}, whereas they can be quadratic in general \cite{KisielewiczSzykula2015SynchronizingAutomataWithExtremalProperties}.
An upper bound $(n-1)(n-2)+1$ on the length of the shortest reset words for Eulerian automata was obtained by Kari~\cite{Kari2003Eulerian}.
Several generalizations of Eulerian automata were proposed: the class of pseudo-Eulerian automata \cite{Steinberg2011AveragingTrick}, for which the same bound $(n-1)(n-2)+1$ was obtained, unambiguous Eulerian automata \cite{CarpiDAllesandro2009StronglyTransitive} for which the \v{C}ern\'{y} bound $(n-1)^2$ was obtained, and quasi-Eulerian automata \cite{Berlinkov2013QuasiEulerianOneCluster}, for which a quadratic upper bound was obtained.
The best lower bound so far was $\frac{1}{2}n^2-\frac{3}{2}n+2$, found by Gusev \cite{Gusev2013LowerBoundsForTheLengthOfResetWordsInEulerianAutomata}.
A series whose shortest reset words seem to have length $\frac{1}{2}n^2-\frac{5}{2}$ was found by Martyugin (unpublished), but no proof has been established.
Further discussion on the bounds for Eulerian automata can be found in the survey~\cite{KariVolkov2013Handbook}.

Here we improve the lower bound by introducing an extremal series of Eulerian automata over a quaternary alphabet with the shortest reset words of length $\frac{1}{2}n^2-\frac{3}{2}$. To prove that, we use a technique of \emph{backward tracing}, which turns out to be very useful in analysis of extremal series of automata in general.
We conjecture that the new lower bound is tight for the class of Eulerian automata. Our exhaustive search over small automata did not find any counterexample.

The new series exhibits the extremal property that some of its subsets require extending words of length exactly $n-1$. This matches the upper bound, which was used in~\cite{Kari2003Eulerian} to obtain the best known upper bound $(n-1)(n-2)+1$ on the length of shortest reset words. Thus, possible improvements of the upper bound require a more subtle method.

\section{Preliminaries}

A \emph{deterministic finite automaton} (\emph{DFA}) is a triple $\mathcal{A} = (Q, \Sigma, \delta)$, where $Q$ is a finite non-empty set of \emph{states}, $\Sigma$ is a finite non-empty \emph{alphabet}, and $\delta\colon Q \times \Sigma \mapsto Q$ is a complete \emph{transition function}. We extend $\delta$ to $Q \times \Sigma^*$ and $2^Q \times \Sigma^*$ as usual.
When $\mathcal{A}$ is fixed, we write shortly $q\cdot w$ and $S\cdot w$ for $\delta(q,w)$ and $\delta(S,w)$ respectively.
The \emph{preimage} of $S\subset Q$ by $w\in \Sigma^*$ is defined as $$\delta^{-1} (S,w)=\{q \in Q\mid q\cdot w \in S\},$$ which is also denoted by $S\cdot w^{-1}$. If $S=\{q\}$ is a singleton, we write $q\cdot w^{-1}$.

A word $w \in \Sigma^*$ is a \emph{reset word} if $|Q \cdot w| = 1$. Note that in this case $Q \cdot w = \{q_0\}$ and $\{q_0\} \cdot w^{-1} = Q$ for some $q_0 \in Q$.
A DFA is called \emph{synchronizing} if it admits a reset word.
The \emph{reset threshold} of a synchronizing DFA $\mathcal{A}$ is the length of the shortest reset words and is denoted by $\rt(\mathcal{A})$.

A word $w$ \emph{extends} a subset $S \subset Q$ if $|S\cdot w^{-1}| > |S|$. In this case we say that $S$ is \emph{$w$-extensible}.

A DFA $\mathcal{A}$ is \emph{Eulerian} if the underlying digraph of $\mathcal{A}$ is strongly connected and the in-degree equals the out-degree for each vertex of the underlying digraph. Equivalently, at every vertex there must be exactly $|\Sigma|$ incoming edges.

We say that a word $w \in \Sigma^*$ is:
\begin{itemize}
\item \emph{permutational} if $Q \cdot w = Q$,
\item \emph{involutory} if $q \cdot w^2=q$ for each $q\in Q$,
\item \emph{unitary} if $p\cdot w \neq p$ holds for exactly one $p \in Q$.
\end{itemize}
Note that each involutory word is permutational.
Also, $w$ is unitary if and only if its action maps exactly one state to another one and fixes all the other states.
For $p,r\in Q$, we write $w = (p \to r)$ if the action of $w\in \Sigma^*$ is defined as $p\cdot w = r$ and $q\cdot w = q$ for each $q\in Q\setminus\{p\}$.

The reversal of a word $w$ is denoted by $w^\mathrm{R}$.
\begin{lemma}\label{pos equal to neg}
Let $\mathcal{A} = (Q, \Sigma, \delta)$ be a DFA. Let $w\in \Sigma^*$ contain only involutory letters. Then $S\cdot w^{-1}=S\cdot w^\mathrm{R}$ for each $S\subseteq Q$.
\end{lemma}
\begin{proof}
If $|w|=0$, the claim is trivial. Inductively, let $w=xv$ for $x\in \Sigma$.
We have $S \cdot (xv)^{-1} = (S \cdot v^{-1})\cdot x^{-1} = (S \cdot v^\mathrm{R})\cdot x^{-1}$ by the inductive assumption, which is equal to $(S \cdot v^\mathrm{R})\cdot x^{-1}\cdot x^2 = (S \cdot v^\mathrm{R})\cdot x$ since $x$ is involutory.
\qed
\end{proof}

\section{Backward Tracing}

There exist several methods of proving reset thresholds of particular series of automata. Here we discuss one of them as a general approach, which we call \emph{backward tracing}.

\begin{definition}
Let $\mathcal{A}$ be a synchronizing DFA and let $u$ be a reset word for $\mathcal{A}$ with $Q\cdot u=\{q_0\}$. We say that $u$ is \emph{straight} if $$q_0 \cdot (u_{\mathbf{m}} u_{\mathbf{s}})^{-1}\not\subseteq q_0 \cdot(u_{\mathbf{s}})^{-1}$$ for each $u_{\mathbf{p}},u_{\mathbf{m}},u_{\mathbf{s}} \in \Sigma^*$ with $u_{\mathbf{p}} u_{\mathbf{m}} u_{\mathbf{s}} = u$.
\end{definition}
The following is a simple observation (cf.~\cite[Theorem 1]{KKS2015ComputingTheShortestResetWords}):
\begin{proposition}\label{prop:ShortestIsStraight}
In a synchronizing DFA each shortest reset word is straight.
\end{proposition}

The observation above leads to a method of proving reset thresholds of particular DFA series by analyzing subsets that are preimages of a singleton under the action of suffixes of length $i=1,2,\ldots,\rt(\mathcal{A})$ of straight reset words.
This works well if the number of such subsets is small in every step, i.e., for each $i$. Note that in general it can grow exponentially.

Interestingly, all known series of most extremal automata, such as
the \v{C}ern\'{y} automata having reset threshold $(n-1)^2$ \cite{Cerny1964}, the twelve known \emph{slowly synchronizing series} having only slightly smaller reset thresholds \cite{AGV2010,AGV2013,KisielewiczSzykula2015SynchronizingAutomataWithExtremalProperties}, and DFAs with cycles of two different lengths \cite{GusevPribavkina2014ResetThresholdsOfAutomataWithTwoCycleLengths},
have the property that the number of possible subsets in each step is bounded by a constant.
We call such series \emph{backward tractable}. 
It is worth mentioning that for such automata we can compute shortest reset words in polynomial time \cite{KKS2015ComputingTheShortestResetWords}.

In this paper, we apply this method to a new series of Eulerian automata, which is backward tractable as well, and whose construction is different from the other known extremal series; in particular, the letters act in many short cycles instead of few large ones.

The new DFAs use only permutational and unitary letters.
This property (which also implies an upper bound $2(n-1)^2$ on the reset threshold \cite{Rystsov2000SimpleIdempotents}) allows us to strengthen the restriction on suffixes to be considered within the backward tracing:
\begin{definition}
With respect to a fixed DFA $\mathcal{A}$, a reset word $u \in \Sigma^*$ with $Q \cdot u = \{q_0\}$ is \emph{greedy}, if for each suffix $v$ of $u$ it holds that: if some $x\in \Sigma$ extends $q_0 \cdot v^{-1}$, then $yv$ is a suffix of $u$ for some $y \in \Sigma$  that extends $q_0 \cdot v^{-1}$.
\end{definition}

\begin{lemma}\label{lem:exists a shortest greedy}
If a synchronizing DFA $\mathcal{A}$ has only permutational and unitary letters, then there exists a shortest reset word that is greedy.
\end{lemma}
\begin{proof}
Let $\Sigma = \Sigma_{\mathrm{p}} \cup \Sigma_{\mathrm{u}}$, where $\Sigma_{\mathrm{p}}$ contains permutational letters and $\Sigma_\mathrm{u}$ contains unitary letters.

Suppose for a contradiction that there is no shortest reset word that is greedy.
Let $u$ be a shortest reset word of $\mathcal{A}$ with the property that its shortest suffix $v$ violating the
greediness is the longest possible.
In other words, the shortest suffix $yv$ of $u$, $y \in \Sigma$, such that some $x\in \Sigma$ extends $q_0 \cdot v^{-1}$ but $y$ doesn't, is the longest possible.

For each suffix $zt$ of $u$ with $z = (p \to q)\in \Sigma_{\mathrm{u}}$, the set $S=q_0 \cdot t^{-1}$ is necessarily $z$-extensible. Indeed, if $q \in S$ and $p \notin S$, then $S$ is clearly $z$-extensible. If $q \notin S$ or $p \in S$, then $S\cdot z^{-1}\subseteq S$, which contradicts Proposition~\ref{prop:ShortestIsStraight}.
Since the inverse actions of the letters from $\Sigma_{\mathrm{p}}$ preserve sizes of subsets, it follows that $u$ contains exactly $|Q|-1$ occurrences of unitary letters.

Write $u = v'v$ and let $u' = v'xv$.
Observe that $u'$ is also a reset word for $\mathcal{A}$:
$q_0 \cdot(xv)^{-1}$ is a (possibly proper) superset of $q_0 \cdot v^{-1}$; hence, we still have $q_0 \cdot(v'xv)^{-1}=Q$.
Since $u'$ contains $|Q|$ occurrences of letters from $\Sigma_{\mathrm{u}}$, and letters from $\Sigma_{\mathrm{p}}$ do not decrease the size of a subset, at least one occurrence of $y \in \Sigma_{\mathrm{u}}$ is not applied to an $y$-extensible subset.
Moreover, this occurrence lies within $v'$, because $v$ is the shortest suffix violating the greediness. Let $v''$ be the word obtained by removing that occurrence of $y$.
We have $u'' = v''xv$, $|u''| = |u|$, and the shortest suffix violating the greediness is longer than $v$.
This yields a contradiction with the choice of $u$.
\qed
\end{proof}

\section{The Extremal Series of Eulerian Automata}

Fix an arbitrary $m\ge1$. Let $N=4m+1$ and $\mathcal{A}_{m}=\left\langle Q,\Sigma,\delta\right\rangle$,
where $Q = \left\{ 0,1,\dots,N-1\right\}$, $\Sigma = \left\{ \alpha,\beta,\omega_{0},\omega_{1}\right\}$. The action of $\alpha$ and $\beta$ is defined by
\begin{eqnarray*}
q\cdot\alpha & = & (-q-1)\bmod N,\\
q\cdot\beta & = & (-q+1)\bmod N,
\end{eqnarray*}
for $q \in Q$, while the action of $\omega_{0}, \omega_{1}$ is defined by
\begin{eqnarray*}
\omega_{0} & = & (1\to 0),\\
\omega_{1} & = & (0\to 1).
\end{eqnarray*}
The automaton $\mathcal{A}_{m}$ is illustrated in Fig.~\ref{fig:A_m}.
We are going to prove that
\[
\rt(\mathcal{A}_{m}) = \frac{N^{2}-3}{2}.
\]
Throughout the proof we use usual operations and inequalities on integers.
Each use of modular arithmetic is described explicitly using the binary
operator \hbox{``$\bmod$''}.
\begin{figure}
\begin{centering}
\includegraphics[scale=1]{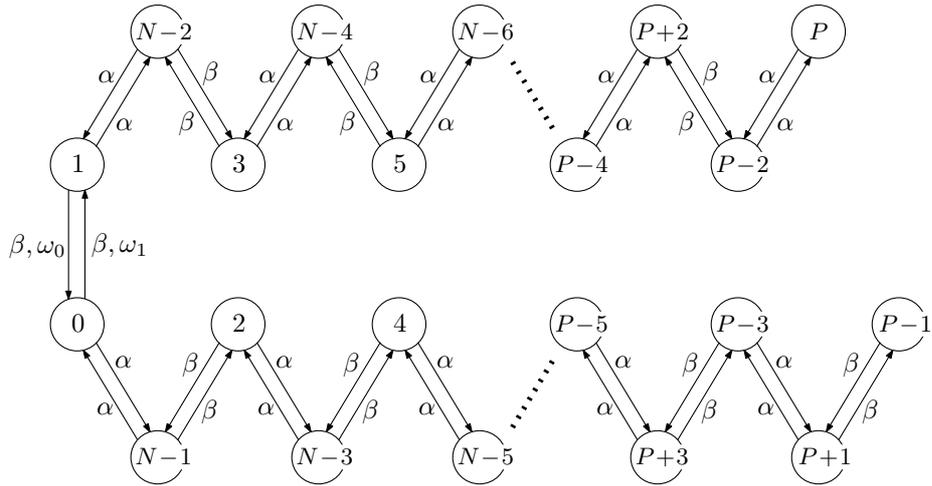}
\par\end{centering}
\caption{The DFA $\mathcal{A}_{m}$, loops are omitted, $P=\frac{N+1}{2}$
\label{fig:A_m}
}
\end{figure}

We use backward tracing to show that there is a unique optimal way to extend a singleton to $Q$.
Note that $\omega_{0}$ and $\omega_{1}$ are unitary, while $\alpha$ and $\beta$ are involutory.
The following notation will be very useful in the analysis of reset words for $\mathcal{A}_{m}$:\\
For $j=0,\ldots,N$ we set:
\begin{center}$\begin{array}{lllclll}
Q_{j} & = & \left\{q \mid 0\le q\le j-1\right\},&\quad&
R_{j} & = & Q_{j}\cdot\beta,\\
Q_{j}^{\diamond} & = & Q_{j}\setminus\left\{ 0\right\} = \left\{q \mid 1 \le q \le j\right\},&\quad&
R_{j}^{\diamond} & = & R_{j}\setminus\left\{ 1\right\} = Q_{j}^{\diamond}\cdot\beta.
\end{array}$\end{center}

\subsection{Construction of a Reset Word}

For an odd $i\ge1$, we define
$$t_{i}=\alpha\left(\beta\alpha\right)^{\frac{i-1}{2}}.$$
Note that:
\begin{enumerate}
\item $\left|t_{i}\right|=i$,
\item $t_{i}$ is a palindrome (i.e., $t_i = t^\mathrm{R}_i$).
\end{enumerate}
By Lemma~\ref{pos equal to neg}, $S \cdot t_{i} = S \cdot t^\mathrm{R}_i = S \cdot t^{-1}_i$ for each $S\subseteq Q$, and we will often interchange $t_{i}^{-1}$ with $t_{i}$. It follows that $qt^2_{i} = qt_{i}t^{-1}_{i} = q$ for every $q \in Q$ and $t_{i}$ is involutory.

\begin{lemma}\label{lem:expl tr}
Let $q\in Q$. It holds that:
\begin{enumerate}
\item\label{enu:kba} $q\cdot\left(\beta\alpha\right)^{h}=(q-2h)\bmod N$
for each $h\ge0$,
\item\label{enu:tNj jk} $q\cdot t_{i}=(-q-i)\bmod N$ for each $i\ge1$.
\end{enumerate}
\end{lemma}
\begin{proof}
The first claim follows trivially from the case of $h=1$.
In this case we have $\left(q\cdot\beta\right)\cdot\alpha=(-\left(-q+1\right)-1)\bmod N=(q-2)\bmod N$.
For the second claim we observe $k\cdot t_{i}=\left(k\cdot\alpha\right)\cdot\left(\beta\alpha\right)^{\frac{i-1}{2}}$, which equals $(-q-1-\left(i-1\right))\bmod N=(-q-i)\bmod N$.
\qed
\end{proof}

\begin{lemma}\label{lem:cikcak list}
Let $2\le j\le N-2$. It holds that:
\begin{enumerate}
\item \label{enu:QjtNj}$Q_{j}\cdot t_{N-j}=Q_{j+1}^{\diamond}$ if $j$
is even,
\item \label{enu:Rjtj}$R_{j}\cdot t_{j-2}=R_{j+1}^{\diamond}$ if $j$
is odd,
\item \label{enu:QtQ}$Q_{j+1}\cdot t_{N-j}=Q_{j+1}$ if $j$ is even,
\item \label{enu:RtR}$R_{j+1}\cdot t_{j-2}=R_{j+1}$ if $j$ is odd.
\end{enumerate}
\end{lemma}
\begin{proof}
For (1) and (2) we use Lemma~\ref{lem:expl tr}(\ref{enu:tNj jk}) with $i=N-j$ and $i=j-2$ respectively and then substitute $d=j-q$:
\begin{eqnarray*}
Q_{j}\cdot t_{N-j} & = & \left\{ j-q\mid q\in Q_{j}\right\} = \left\{ d\mid1\le d\le j\right\} = Q_{j+1}^{\diamond},\\
R_{j}\cdot t_{j-2} & = & \left\{ q\cdot\beta t_{j-2}\mid q\in Q_{j}\right\} = \left\{(-(-q+1)-\left(j-2\right))\bmod N\mid q\in Q_{j}\right\} =\\
 & = & \left\{(q-j+1)\bmod N\mid q\in Q_{j}\right\} = \left\{ (-d+1)\bmod N\mid1\le d\le j\right\} =\\
 & = & \left\{ d\cdot\beta\mid1\le d\le j\right\} =R_{j+1}^{\diamond}.
\end{eqnarray*}
For (3) and (4) we use (\ref{enu:QjtNj}) and (\ref{enu:Rjtj}) respectively and the fact that $t_{N-j}$ and $t_{j-2}$ are involutory. We have:
\begin{eqnarray*}
Q_{j+1}\cdot t_{N-j} & = & Q_{j+1}^{\diamond}\cdot t_{N-j}\cup\left\{ 0\cdot t_{N-j}\right\} =Q_{j}\cup\left\{ j\right\} =Q_{j+1},\\
R_{j+1}\cdot t_{j-2} & = & R_{j+1}^{\diamond}\cdot t_{j-2}\cup\left\{ 1\cdot t_{j-2}\right\} =R_{j}\cup\left\{ N-j+1\right\} \\
 & = & R_{j}\cup\left\{ j\cdot\beta\right\} = R_{j+1}.
\end{eqnarray*}
\qed
\end{proof}

\noindent Let
\begin{equation*}\label{eq:def of w}
w=v_{N-1}\beta v_{N-2}\beta v_{N-3}\dots\beta v_{3}\beta v_{2},
\end{equation*}
where
\[
v_{j}=\begin{cases}
\omega_{1}t_{N-j} & \mbox{if }j\mbox{ is even},\\
\omega_{0}t_{j-2} & \mbox{if }j\mbox{ is odd}.
\end{cases}
\]
In Lemma~\ref{lem:analysis of w} below, we show that $w$ extends $Q_2$ to $Q_N$ according to the following scheme:
\[
\begin{array}{c}
Q_{2}\overset{v_{2}^{-1}}{\mapsto}Q_{3}\overset{\beta^{-1}}{\mapsto}R_{3}\overset{v_{3}^{-1}}{\mapsto}R_{4}\overset{\beta^{-1}}{\mapsto}Q_{4}\overset{v_{4}^{-1}}{\mapsto}Q_{5}\overset{\beta^{-1}}{\mapsto}R_{5}\overset{v_{5}^{-1}}{\mapsto}R_{6}\overset{\beta^{-1}}{\mapsto}Q_{6}\mapsto\cdots\\
\cdots\mapsto Q_{N-3}\overset{v_{N-3}^{-1}}{\mapsto}Q_{N-2}\overset{\beta^{-1}}{\mapsto}R_{N-2}\overset{v_{N-2}^{-1}}{\mapsto}R_{N-1}\overset{\beta^{-1}}{\mapsto}Q_{N-1}\overset{v_{N-1}^{-1}}{\mapsto}Q_{N},
\end{array}
\]
and thus the word $w\omega_0$ is a reset word for $\mathcal{A}_m$.
\begin{remark}\label{enu:po alfa je beta}
The word $w$ ends with $\alpha$. The other occurrences of $\alpha$ in $w$ are directly followed by $\beta$.
\end{remark}
\begin{remark}
A set $S\subseteq Q$ is:
\begin{itemize}
\item \label{enu:o2}$\omega_{0}$-extensible if and only if $S\cap\left\{ 0,1\right\} =\left\{ 0\right\} $,
\item \label{enu:o1}$\omega_{1}$-extensible if and only if  $S\cap\left\{ 0,1\right\} =\left\{ 1\right\} $.
\end{itemize}
\end{remark}
We say that a set $S\subseteq Q$ is \emph{$\omega$-extensible} if it is $\omega_{0}$-extensible or $\omega_{1}$-extensible.

\begin{lemma}\label{lem:analysis of w}
Let $2\le j\le N-1$. It holds that:
\begin{enumerate}
\item $Q_{2}\cdot\left(v_{j}\beta v_{j-1}\dots\beta v_{2}\right)^{-1}=Q_{j+1}$
if $j$ is even,
\item $Q_{2}\cdot\left(v_{j}\beta v_{j-1}\dots\beta v_{2}\right)^{-1}=R_{j+1}$
if $j$ is odd,
\item \label{enu:w je reset}$w\omega_{0}$ is a reset word of $\mathcal{A}_{m}$.
\end{enumerate}
\end{lemma}
\begin{proof}
We prove the first two claims by induction.
For $j=2$, using Lemma~\ref{lem:cikcak list}(\ref{enu:QjtNj}) we have:
$$Q_{2}\cdot v_{2}^{-1} = (Q_{2}\cdot t_{N-2}^{-1})\cdot\omega_{1}^{-1} = Q_{3}^{\diamond}\cdot\omega_{1}^{-1}=Q_{3}.$$
Next, take $j\ge2$ and suppose that both the claims hold for $j-1$.
We use the induction hypothesis and, depending on the parity of $j$, Lemma~\ref{lem:cikcak list}(\ref{enu:QjtNj}) or Lemma
\ref{lem:cikcak list}(\ref{enu:Rjtj}) respectively.
For an even $j$ we have:
\begin{eqnarray*}
Q_{2}\cdot\left(v_{j}\beta v_{j-1}\dots\beta v_{2}\right)^{-1} & = & R_{j}\cdot\left(v_{j}\beta\right)^{-1}=Q_{j}\cdot v_{j}^{-1} = Q_{j}\cdot\left(\omega_{1}t_{N-j}\right)^{-1} =\\
& = & Q_{j+1}^{\diamond}\cdot\omega_{1}^{-1}=Q_{j+1},
\end{eqnarray*}
and for an odd $j$ we have: 
\begin{eqnarray*}
Q_{2}\cdot\left(v_{j}\beta v_{j-1}\dots\beta v_{2}\right)^{-1} & = & Q_{j}\cdot\left(v_{j}\beta\right)^{-1}=R_{j}\cdot v_{j}^{-1}=R_{j}\cdot\left(\omega_{0}t_{j-2}\right)^{-1}=\\
 & = & R_{j+1}^{\diamond}\cdot\omega_{0}^{-1}=R_{j+1}.
\end{eqnarray*}
The claim~(\ref{enu:w je reset}) follows from $Q_{1}\cdot\left(w\omega_{0}\right)^{-1}=Q_{2}\cdot w^{-1}=Q_{N}$, according to the first claim with $j=N-1$.
\qed
\end{proof}

\noindent It remains to calculate the length of $w$.
\begin{lemma}
The length of $w$ is $\frac{N^{2}-5}{2}$.
\end{lemma}
\begin{proof}
The sum of $\left|v_{i}\right|$
with even $i$ is 
\[
\sum_{i=1}^{\frac{N-1}{2}}\left(1+N-2i\right)=\frac{\left(N-1\right)\left(1+N\right)}{2}-\frac{\left(N-1\right)\left(1+N\right)}{4}=\frac{1}{4}\left(N^{2}-1\right),
\]
and the sum of $\left|v_{i}\right|$ with odd $i$ is 
\[
\sum_{i=1}^{\frac{N-3}{2}}2i=\frac{\left(N-3\right)\left(N-1\right)}{4}=\frac{1}{4}\left(N^{2}-4N+3\right).
 \]
\end{proof} 
Together with the $N-3$ occurrences of $\beta$, we have $\left|w\right|=\frac{N^{2}-5}{2}$.
\qed
\noindent Thus, we have that $w\omega_0$ is a reset word for $\mathcal{A}_m$ with length $\left|w\omega_0\right|=\frac{N^{2}-3}{2}$.

\subsection{Lower Bound on the Reset Threshold}

Finally, let us show that no reset word for $\mathcal{A}_k$ is shorter than $w\omega_0$.

\begin{lemma}\label{lem:factors o0B,o1B}
If $v\in \Sigma^*$ is greedy and straight reset word with $Q\cdot v= \{q_0\}$, then $v$ does not contain $\omega_{0}\beta$ nor $\omega_{1}\beta$ as a factor.
\end{lemma}
\begin{proof}
Suppose for a contradiction that $v=u''x\beta u'$ for $x \in \{\omega_{0},\omega_{1}\}$. Since $v$ is greedy, $\{q_0\}\cdot (u')^{-1}$ is not $\omega$-extensible, so it contains both $0$ and $1$ or neither of them. Since $\beta$ switches these states, $\{q_0\}\cdot (\beta u')^{-1}$ has the same property. Then $\{q_0\}\cdot (\beta u')^{-1} = \{q_0\}\cdot (x\beta u')^{-1}$ and so $v$ is not straight.
\qed
\end{proof}

\begin{lemma}\label{lem:w greedy}
Let $2\le j\le N-1$.
It holds that:
\begin{enumerate}
\item\label{enu:Qjth} $\left\{ 0,1\right\} \cap (Q_{j}\cdot t_{h})=\emptyset$
for $1\le h<N-j$ if $j$ is even,
\item\label{enu:Rjth} $\left\{ 0,1\right\} \subseteq (R_{j}\cdot t_{h})$\emph{
for $1\le h<j-2$ if $j$ is odd.}
\end{enumerate}
\end{lemma}
\begin{proof}
As $t_{h}$ is involutory, it is enough to show
for $q \in \{0,1\}$ that $q\cdot t_{h}\notin Q_{j}$ or $q\cdot t_{h}\in R_{j}$ respectively.

As for~(\ref{enu:Qjth}), by Lemma~\ref{lem:expl tr}(\ref{enu:tNj jk}) we have $q\cdot t_{h}=N-q-h>j-1$, thus $q\cdot t_{h}\not\in Q_{j}$.
As for (\ref{enu:Rjth}), denoting $q'=q\cdot t_{h}$, we have $q'=N-q-h$.
Then $q'\cdot\beta = (q+h+1)\bmod N$, and since $q\le1$ and $h<j-2$, we get $q'\cdot\beta <j$, which implies $q'\cdot\beta\in Q_{j}$ and $q'\in R_{j}$.
\end{proof}

\begin{lemma}\label{enu:no idle omega}
For each suffix $xu$ of $w$ with $x\in\left\{ \omega_{0},\omega_{1}\right\}$
and $u\in\Sigma^{*}$ it holds that $x$ extends $Q_{2}\cdot u^{-1}$.
\end{lemma}
\begin{proof}
For every suffix $\omega_{1}u$ we have $Q_{2}\cdot u^{-1}\omega_{1}^{-1} = Q_{j+1}^{\diamond}\cdot\omega_{1}^{-1}=Q_{j+1}$ for some even $j$, and for every suffix $\omega_{0}u$ we have $Q_{2}\cdot u^{-1}\omega_{0}^{-1} = R_{j+1}^{\diamond}\cdot\omega_{0}^{-1}=R_{j+1}$ for some odd $j$.
\qed
\end{proof}
\begin{lemma}\label{lem:w0 je greedy}
The word $w\omega_{0}$ is greedy.
\end{lemma}
\begin{proof}
Let $u$ be the shortest suffix of $w\omega_{0}$ that violates the greediness,
i.e., suppose that $Q_{1}\cdot u^{-1}$ is $z$-extensible for $z \in \{\omega_{0},\omega_{1}\}$, but $zu$ is not a suffix of $w\omega_{0}$. This simplification works because $Q_{1}\cdot u^{-1}$ cannot be both $\omega_0$-extensible and $\omega_1$-extensible. 
Fix $x \in \Sigma$ such that $xu$ is a suffix of $w$.
Let $u=y u_{\mathbf{s}}$ with $y \in \Sigma$.

If $y \in \{\omega_{0},\omega_{1}\}$ then $Q_{1}\cdot (y u_{\mathbf{s}})^{-1}$ is not $\omega$-extensible. If $x \in \{\omega_0,\omega_1\}$, then $Q_{1}\cdot (y u_{\mathbf{s}})^{-1}$ is $x$-extensible due to Lemma~\ref{enu:no idle omega}. Thus, necessarily $x,y \in \{\alpha,\beta\}$.

Assume $y=\beta$.
If $Q_{1}\cdot\left(y u_{\mathbf{s}}\right)^{-1}$ is $\omega$-extensible, then $Q_{1}\cdot(u_{\mathbf{s}})^{-1}$ is $\omega$-extensible as well due to $0\cdot\beta=1$ and $1\cdot\beta=0$, implying that $u_{\mathbf{s}}$ is a shorter suffix violating the greediness.

Assume $y=\alpha$.
Because $w$ does not contain the factor $\alpha\alpha$, it follows that $x=\beta$. According to~(\ref{eq:def of w}), i.e., the definition of $w$, and
the fact that $v_{N-1}=\alpha$, the factor $xy=\beta\alpha$ occurs only
within the factors $v_{2}\dots,v_{N-2}$. Thus,
\[
yu_{\mathbf{s}}=\alpha\left(\beta\alpha\right)^{i}\beta\left(v_{j-1}\beta v_{j-2}\dots\beta v_{3}\beta v_{2}\right)\omega_{0},
\]
where $\alpha\left(\beta\alpha\right)^{i}$ is a suffix of $v_{j}$. We apply Lemma \ref{lem:analysis of w}:
\begin{enumerate}
\item If $j$ is odd, we get $Q_{2}\cdot\left(v_{j-1}\beta v_{j-2}\dots\beta v_{2}\right)^{-1}=Q_{j}$,
while  $v_{j}=\omega_{0}t_{j-2}$ and $i\le\frac{j-3}{2}$. Then
\[
Q_{1}\cdot\left(y u_{\mathbf{s}}\right)^{-1}=Q_{j}\cdot\left(\alpha\left(\beta\alpha\right)^{i}\beta\right)^{-1}=Q_{j}\cdot\left(t_{h}\beta\right)^{-1}=R_{j}\cdot t_{h}^{-1}=R_{j}\cdot t_{h},
\]
where $h=2i+1$.
We see that $1\le h\le j-2$.
If $h=j-2$, then $\omega_{0}t_{h}=v_{j}$, so $x=\omega_{0}$.
Otherwise we apply Lemma~\ref{lem:w greedy}(\ref{enu:Rjth})
to get $\left\{0,1\right\} \subseteq R_{j}\cdot t_{h}$, which contradicts
that $Q_{1}\cdot\left(y u_{\mathbf{s}}\right)^{-1}$ is $\omega$-extensible.

\item If $j$ is even, we get $Q_{2}\cdot\left(v_{j-1}\beta v_{j-2}\dots\beta v_{2}\right)^{-1}=R_{j}$,
while $v_{j}=\omega_{1}t_{N-j}$ and $i\le\frac{N-j-1}{2}$. Then
\[
Q_{1}\cdot\left(y u_{\mathbf{s}}\right)^{-1}=R_{j}\cdot\left(\alpha\left(\beta\alpha\right)^{i}\beta\right)^{-1}=R_{j}\cdot\left(t_{h}\beta\right)^{-1}=Q_{j}\cdot t_{h},
\]
where $h=2i+1$.
We see that $1\le h\le N-j$.
If $h=N-j$, then $\omega_{1}t_{h}=v_{j}$, so $x=\omega_{1}$.
Otherwise we apply Lemma~\ref{lem:w greedy}(\ref{enu:Qjth}) to get $\left\{ 0,1\right\} \cap Q_{j}\cdot t_{h}=\emptyset$, which contradicts that $Q_{1}\cdot\left(y u_{\mathbf{s}}\right)^{-1}$ is $\omega$-extensible.\qed
\end{enumerate}
\end{proof}

\begin{lemma}\label{lem:exists rw ending with}
There exists a shortest reset word for $\mathcal{A}_{m}$ that ends with $\omega_{0}$ and is greedy.
\end{lemma}
\begin{proof}
Lemma~\ref{lem:exists a shortest greedy} gives a shortest reset word that is greedy. Clearly, a shortest reset word ends with a non-permutational letter, i.e., $\omega_0$ or $\omega_1$. In the latter case, replacing the ending $\omega_1$ with $\omega_0$ yields a reset word of the same length and preserves greediness.
\qed
\end{proof}

\begin{theorem}\label{thm:w je shortest}
The word $w\omega_{0}$ is a shortest reset word for $\mathcal{A}_m$.
\end{theorem}
\begin{proof}
Using Lemma~\ref{lem:exists a shortest greedy}, let $w'\omega_{0}$ be a greedy shortest reset word of $\mathcal{A}_{m}$.
If $w' = w$, we are done, so let $w'\neq w$ and let $w_{\mathbf{s}}$ be the longest common suffix of $w'$ and $w$.

If $w_{\mathbf{s}}=w'$, then $w'$ is a proper suffix of $w$ and so it contains at most $N-3$ letters from $\{\omega_0,\omega_1\}$, which contradicts that $w_{\mathbf{s}}\omega_0$ is a reset word.
So we can write $w=w_{\mathbf{p}}xw_{\mathbf{s}}$ and
$w'=w'_{\mathbf{p}}x'w_{\mathbf{s}}$,
where $x,x'\in\Sigma$ and $x'\neq x$.
We will show that each of the following cases according to $x$ and $x'$ leads to a contradiction:
\begin{enumerate}
\item Suppose that $x\in\left\{\omega_{0},\omega_{1}\right\}$.
Then Lemma~\ref{enu:no idle omega} implies that $Q_{2}\cdot w_{\mathbf{s}}$ is $x$-extensible, which contradicts $x'\neq x$ and the greediness of $w'\omega_{0}$.
\item Suppose that $x'\in\left\{\omega_{0},\omega_{1}\right\}$. According to Proposition~\ref{prop:ShortestIsStraight}, $w'\omega_{0}$ is straight, which implies that $Q_{2}\cdot w_{\mathbf{s}}$ is $x'$-extensible, which contradicts $x'\neq x$ and Lemma~\ref{lem:w0 je greedy}, i.e, the greediness of $w\omega_{0}$.
\item Suppose that $x=\alpha$ and $x'=\beta$.
According to Remark~\ref{enu:po alfa je beta}, $w_{\mathbf{s}}=\epsilon$ or $w_{\mathbf{s}}$ starts with $\beta$.
 The case of $w_{\mathbf{s}}=\epsilon$ contradicts the straightness of $w'\omega_{0}$ because each $x'\in\Sigma\setminus\{\alpha\}$ satisfies $Q_{2}\cdot\left(x'\right)^{-1}=Q_{2}$. The other case implies $\beta\beta$ occurring in $w'$ and thus also contradicts the straightness of $w'\omega_{0}$.
\item Suppose that $x=\beta$ and $x'=\alpha$.
Then $w_{\mathbf{s}}\neq\epsilon$. If $w_{\mathbf{s}}$ starts with $\alpha$ or $\beta$, then either $w'$ or $w$ contains the factor $\alpha\alpha$
or $\beta\beta$, which contradicts the straightness of $w'\omega_{0}$ or the definition of $w$.
Hence, $w_\mathbf{s}$ starts with $\omega_{0}$ or $\omega_{1}$. Since this starts a factor $v_{j}$ for some $j \ge 2$, we can write 
\[
w_{\mathbf{s}}=v_{j}\beta v_{j-1}\dots\beta v_{3}\beta v_{2}.
\]
We consider the following two subcases:

\begin{enumerate}
\item Suppose that $w_{\mathbf{s}}$ starts with $\omega_{1}$.
Note that $j\ge 2$ is even and $Q_{2}\cdot w_{\mathbf{s}}^{-1} = Q_{j+1}$ by Lemma~\ref{lem:analysis of w}.
Let $w_{\mathbf{m}}$ be the longest common suffix of $w_{\mathbf{p}}'x'=w_{\mathbf{p}}'\alpha$ and $t_{N-j}$.
Clearly, $\left|w_{\mathbf{m}}\right|\ge1$. If $w_{\mathbf{m}} = t_{N-j}$, then from Lemma~\ref{lem:cikcak list}(\ref{enu:QtQ}) we have $Q_{j+1}\cdot t_{N-j}=Q_{j+1}$, which contradicts the straightness of $w'\omega_{0}$.
If $w_{\mathbf{m}}= w_{\mathbf{p}}'x'$, then $w'$ starts with $\alpha$ or $\beta$, which contradicts that $w'\omega_{0}$ is a shortest reset word.
It follows that we can write $w'=w_{\mathbf{pp}}'y'w_{\mathbf{m}}w_{\mathbf{s}}$ for $y'\in\Sigma$.
Moreover, as $w'$ does not contain the factors $\alpha\alpha$ and $\beta\beta$, we have $y'\neq\alpha$ and $y'\neq\beta$, so $y' \in \{\omega_{0},\omega_{1}\}$.
Due to Lemma~\ref{lem:factors o0B,o1B}, $w_{\mathbf{m}}$ cannot start with $\beta$, and from the construction of $t_{N-j}$ we have $w_{\mathbf{m}}=t_{h}$ for $h\le N-j-2$.
It holds that $Q_{j+1}\cdot w_{\mathbf{m}}^{-1}=Q_{j+1}\cdot t_{h}=Q_{j}\cdot t_{h}\cup\left\{j\cdot t_{h}\right\}$.
Lemma~\ref{lem:w greedy}(\ref{enu:Qjth}) provides that $\left\{0,1\right\} \cap Q_{j}\cdot t_{h}=\emptyset$.
Also, $j\cdot t_{h}=N-j-h\ge2$.
Together, $Q_{j+1}\cdot w_{\mathbf{m}}^{-1}\cap\left\{ 0,1\right\} = \emptyset$, and thus this set is not $\omega$-extensible, which contradicts $y' \in \{\omega_{0},\omega_{1}\}$ and the straightness of $w'\omega_{0}$.

\item Suppose that $w_{\mathbf{s}}$ starts with $\omega_{0}$.
Note that $j\ge3$ is odd and $Q_{2}\cdot w_{\mathbf{s}}^{-1} = R_{j+1}$ by Lemma~\ref{lem:analysis of w}.
Let $w_{\mathbf{m}}$ be the longest common suffix of $w_{\mathbf{p}}'x'=w_{\mathbf{p}}'\alpha$ and $t_{j-2}$.
Clearly, $\left|w_{\mathbf{m}}\right|\ge1$. If $w_{\mathbf{m}} = t_{j-2}$, then from Lemma~\ref{lem:cikcak list}(\ref{enu:RtR}) we have $R_{j+1}\cdot t_{j-2}=R_{j+1}$, which contradicts the straightness of $w'\omega_{0}$. If $w_{\mathbf{m}}= w_{\mathbf{p}}'x'$, then $w'$ starts with $\alpha$ or $\beta$, which contradicts that $w'\omega_{0}$ is a shortest reset word. It follows that we can write $w'=w_{\mathbf{pp}}'y'w_{\mathbf{m}}w_{\mathbf{s}}$ for $y'\in\Sigma$.
Moreover, as $w'$ does not contain the factors $\alpha\alpha$ and $\beta\beta$, we have $y'\neq\alpha$ and $y'\neq\beta$, so $y' \in \{\omega_{0},\omega_{1}\}$.
Due to Lemma~\ref{lem:factors o0B,o1B}, $w_{\mathbf{m}}$ cannot start with $\beta$, and from construction of $t_{j-2}$ we have $w_{\mathbf{m}}=t_{h}$ for $h\le j-4$.
We have $R_{j+1}\cdot w_{\mathbf{m}}^{-1}=R_{j+1}\cdot t_{h}\supseteq R_{j}\cdot t_{h}$.
Lemma~\ref{lem:w greedy}(\ref{enu:Rjth}) gives $\left\{ 0,1\right\} \subseteq R_{j}\cdot t_{h}$.
Thus, $\left\{ 0,1\right\} \subseteq R_{j+1}\cdot w_{\mathbf{m}}^{-1}$, and thus this set is not $\omega$-extensible, which contradicts $y' \in \{\omega_{0},\omega_{1}\}$ and the straightness of $w'\omega_{0}$. \qed
\end{enumerate}
\end{enumerate}
\end{proof}

\noindent Theorem~\ref{thm:w je shortest} implies that $\rt(\mathcal{A}_m) =|w \omega_0| = \frac{N^{2}-3}{2}$.

\subsection{Extending Words}

The general upper bound $(n-2)(n-1)+1$ for reset thresholds of synchronizing Eulerian DFAs comes from the fact that any proper and non-empty subset of $Q$ is extended by a word of length at most $n-1$ \cite{Kari2003Eulerian}, while in the general case the minimum length of extending words can be quadratic (this was shown recently -- see \cite{KisielewiczSzykula2015SynchronizingAutomataWithExtremalProperties}).
In view of this, our series shows that this bound is tight for infinitely many $n$, and so the upper bound for reset thresholds for this class cannot be improved only by reducing this particular bound. The following remark follows from the analysis in the proof of Theorem~\ref{thm:w je shortest}:
\begin{remark}
The shortest extending word of $\{0,1\}$ in $\mathcal{A}_m$ is $v_2 = \omega_1 \alpha (\beta \alpha)^{(N-3)/2}$ of length $N-1$.
\end{remark}


\section{Experiments}

Using the algorithm from~\cite{KS2013GeneratingSmallAutomata,KKS2016ExperimentsWithSynchronizingAutomata}, we have performed an exhaustive search over small synchronizing Eulerian DFAs.
We verified the bound $(n^2-3)/2$ for the case of binary DFAs with $n \le 11$ states, automata with four letters and $n \le 7$ states, DFAs with eight letters and $n \le 5$ states, and all DFAs with $n \le 4$ states.

For $n \in \{3,4,5,7\}$ the bound $(n^2-3)/2$ is reachable.
For $n=7$, up to isomorphism, we identified 2 ternary examples and 12 quaternary examples which also meet the bound.
It seems that our series $\mathcal{A}_m$ is not unique meeting the bound, as some of the quaternary examples could be generalizable to series with the same reset thresholds.
Also, for the binary case we found that for $n \in \{5,7,8,9,11\}$ the bound $(n^2-5)/2$ is met uniquely by DFAs from the Martyugin's series, but it is not reachable for $n \in \{6,10\}$.

\begin{conjecture}
For $n \ge 3$, $(n^2-3)/2$ is an upper bound for the reset threshold of an $n$-state Eulerian synchronizing automaton.
If $|\Sigma|=2$, then the bound can be improved to $(n^2-5)/2$.
\end{conjecture}

\bibliographystyle{splncs03}

\end{document}